\documentclass[10pt, conference]{IEEEtran}
\usepackage[english]{babel}

\usepackage[cmex10]{amsmath}
\usepackage{amssymb}

\usepackage{url}
\usepackage{amsfonts}
\usepackage{graphicx}
\usepackage{subfigure}
\usepackage{algorithm}
\usepackage{algpseudocode}
\newtheorem{theorem}{Theorem}
\newtheorem{proposition}{Proposition}
\newtheorem{remark}{Remark}
\newtheorem{lemma}[theorem]{Lemma}
\newtheorem{corollary}[theorem]{Corollary}
\newtheorem{definition}{Definition}

\newtheorem{example}{Example}
\newtheorem{problem}{Problem}

\newenvironment{proof}
	{
          \noindent {\bf Proof.} 
        }{$\Box$ \vspace{6pt}}

\newcommand{\cbs}{\textsf{CBS}}

\newcommand{\al}[2]{\ensuremath{\alpha \left( #1, \, #2 \right)}}
\newcommand{\be}[2]{\ensuremath{\beta  \left( #1, \, #2 \right)}}

\newcommand{\Prob}[1]{\mathbf{Pr}\left\{ #1 \right\}}
\newcommand{\Exp}[1]{\mathbf{E}\left\{ #1 \right\}}

\newcommand{\natu}{\mathbf{N}}
\newcommand{\real}{\mathbf{R}}

\newcommand{\gm}[1]{\gamma_{#1}}

\usepackage{color}

\IEEEoverridecommandlockouts

\begin{document}
\title{An Analytical Solution for Probabilistic Guarantees of
  Reservation Based Soft Real--Time Systems}
\author{\IEEEauthorblockN{Luigi Palopoli${^1}$, Daniele
    Fontanelli${^2}$, Luca Abeni${^1}$ Bernardo Villalba Fr\'{i}as${^1}$}
  \IEEEauthorblockA{${^1}$Dipartimento di Scienza e Ingegneria
    dell'Informazione\\
    ${^2}$Dipartimento di Ingegneria Industriale\\
    University of Trento,
    Trento, Italy\\
    \{luigi.palopoli,daniele.fontanelli,luca.abeni,br.villalbafrias\}@unitn.it }
  \thanks{The research leading to these results has received funding
    from the European Union FP7 Programme
    (FP7/2007-2013) under grant agreement n$^\circ$ ICT-2011-288917
    ``DALi - Devices for Assisted Living'' and under grant agreement
    $n^\circ$ FP7-ICT-257462 ``HYCON2 NoE'', and from the European
    Union H2020 programme under grant agreement $n^\circ$ 643544 ``ACANTO''} }

\maketitle
\begin{abstract}
  We show a methodology for the computation of the probability of
  deadline miss for a periodic real--time task scheduled by a resource
  reservation algorithm.  We propose a modelling technique for the
  system that reduces the computation of such a probability to that of
  the steady state probability of an infinite state Discrete Time
  Markov Chain with a periodic structure. This structure is exploited
  to develop an efficient numeric solution where different
  accuracy/computation time trade--offs can be obtained by operating on
  the granularity of the model. More importantly we offer a closed
  form conservative bound for the probability of a deadline miss.  Our
  experiments reveal that the bound remains reasonably close to the
  experimental probability in one real--time application of
  practical interest. When this bound is used for the optimisation of the
  overall Quality of Service for a set of tasks sharing the CPU, it
  produces a good sub--optimal solution in a small amount of time.
\end{abstract}


\section{Introduction}
The term {\em soft real--time} is used for a class of
real--time applications that are resilient to
occasional and controlled timing faults. Significant examples include
multimedia streaming~\cite{Abe98-2}, computer vision and real--time
control~\cite{FontanelliGP13,Cer04}. 

An effective method to express the timing requirements for
a soft real--time application is by associating each deadline with a
probability that it will be met: the notion of
\emph{probabilistic deadlines}~\cite{Abe98-3}.  Probabilistic
deadlines can be related to the
Quality of Service (QoS) delivered by the application~\cite{psnr-tools, FontanelliGP13} and, more
generally, enable the expression of a wide
range of performance requirements, where classic hard real--time
systems can be regarded as a special case.

In traditional hard real--time applications, the use of fixed or
dynamic scheduling priorities has gained an undisputed
prominence.  Part of the reasons of this
success is in the presence of efficient numeric techniques that make for
the provision of tight conditions for temporal guarantees~\cite{Pan86}. 
At least as important is a group of approximate analytical results. The most famous is the
utilisation bound~\cite{Liu73},
which offers clear guidelines on how to tweak periods and computation
times in order to meet the deadlines of all tasks in the system.

The use of scheduling priorities allows the designer to define a
partial order between all the tasks in a set and inevitably couples
their timing behaviour. This is acceptable if the purpose is to offer
guarantees for the set as a whole. On the contrary, if the designer
requires specific QoS levels for each task,
scheduling priorities can be too coarse a tool.  For this reason an
intense research work has produced alternative scheduling solutions
for soft real--time systems.  One of the most popular is the
\emph{Resource Reservations} scheduling (RR)~\cite{Raj98,Abe98-2}, which
enables a fine grained control on the fraction of computing power
(bandwidth) that each task receives.   A
key property of RR scheduling is \emph {temporal isolation}:
the ability for a task to meet its deadlines solely depends on its
computation requirement and on its scheduling parameters.  This
property enables the provision of specific temporal guarantees to each
task and simplifies system design. RR scheduling is now
available in the mainstream Linux
Kernel\footnote{\url{https://www.kernel.org/doc/Documentation/scheduler/sched-deadline.txt}}.

When the probability distribution of inter--arrival time and of
computation time are known independent identically distributed
(i.i.d.) stochastic processes, temporal isolation allows modelling the evolution of a task scheduled
through a RR as a Discrete--Time
Markov Chain (DTMC) with an infinite number of
states~\cite{Abe98-3,Abe01-1}. 
In this paper, we restrict the focus to the analysis of periodic tasks.
For this case, we can see that the DTMC describing the system
takes the form of a Quasi--Birth--Death Process (QBDP)~\cite{bini2005numerical}. We introduce a
granularity parameter that allows us to reduce the complexity of the
model at the expense of a conservative approximation in the computation of
the probability. 
We show a novel analysis that exploits the specific structure of the
transition matrix of this QBDP. The outcome is an expression for the
steady state probability of meeting the deadline, which can be used in
different ways. The first one is for the construction of a numeric
algorithm for probabilistic guarantees, with a performance comparable
to the best state of the art techniques for numeric solutions of
QBDP. The second one, the most important, is for the computation of an
analytical conservative bound for the probability of meeting the
deadline.  This bound proves itself reasonably accurate for a large
set of synthetic test cases.  We have also performed a large
collection of experimental data for a real--life application, in which
the presence of several non--idealities (OS overhead, correlation in
the computation times, etc.) challenges the assumptions the method
relies on. The small approximation error that we observed in the
experiments suggests the practical applicability of the method at
least in the considered scenario.  The application of the
bound is very convenient when solving QoS optimisation problems that
require to efficiently identify the minimum bandwith required for
a desired probability of deadline miss. We show a
realistic example of this kind where the application of the analytic
bound produces a good sub--optimal solution in a tiny fraction of the
time required by a numeric approach.

 The paper is organised as follows. In Section~\ref{sec:relwork}, we
 offer a brief survey of the related work. In
 Section~\ref{sec:problem} we formally describe the problem addressed
 in the paper.  In Section~\ref{sec:model},we show how a resource
 reservation can be conservatively modelled as a QBDP.  The
 computation of our analytical bound is reported in
 Section~\ref{sec:analytical}. In
 Section~\ref{sec:experiments}, we prove the validity of the bound in
 a large set of experiments. In Section~\ref{sec:example}, we show the
 concrete application of the method to a QoS optimisation
 problem. Finally, in Section~\ref{sec:conclusions} we offer our
 conclusions and announce the future work directions.

\section{Related Work}
\label{sec:relwork}

The stochastic analysis of performance of soft real--time tasks
started two decades ago.  The same task model presented in this paper
(a triple of period, probability distribution of the task computation
time and requested probability of deadline miss in the long run)
has been also adopted in
the statistical rate monotonic approach~\cite{Atl98}.
More recently, an important number of research papers has concentrated
on the computation of the response time of systems with fixed or
dynamic priority when tasks have stochastic variability in computation
times~\cite{Dia03,Dia04, DBLP:conf/rtss/MaximC13}, in the
inter--arrival time~\cite{Cuc06} or in both~\cite{Gior07}.  Similar
techniques have recently been applied to multiprocessor
systems~\cite{Mil10}.  An obvious point of differentiation between our
technique and the ones describes so far is that while these papers
propose numeric techniques, we offer an analytic bound that is
satisfactorily tight in many cases of interest. A very interesting
connection can be established with the work of Diaz et
al.~\cite{Dia03}, where the authors propose the exact solution for a
specific numeric example. Our computation, on the contrary, applies to
general cases.
What is more, all the approaches mentioned above analyse the task set as a
whole, since real--time schedulers do not enjoy
temporal isolation. This makes QoS optimisation much
more difficult than in our case.

Other authors have analysed scheduling approaches other than
``traditional'' fixed or dynamic priorities. Dong-In et
al.~\cite{Don97} have analysed Time Division Multiple Access (TDMA)
approaches, Haman et al.~\cite{Ham01} have focused on a model where
tasks are split in mandatory and optional parts.  This paper is based
on reservation--based scheduling~\cite{Abe98-2,Raj98}, which allows us
to exploit temporal isolation and analyse each task separately.
Abeni and Buttazzo proposed a model for RR scheduling based
on queueing theory~\cite{Abe98-3,Abe01-1}.  The computation of the
deadline miss probability requires to numerically solve an eigenvector
problem for an infinitely large matrix.  Recently, approximated solution
techniques have been proposed for efficient numeric computation of a
bound for the probability of meeting the deadline~\cite{Ref10}.

In this paper, we show how the adoption of the reservation scheduler
and the restriction to periodic tasks produces a model that is a
particular instance of a QBDP.  Efficient numeric solutions for QBDP
and for M/G/1 queue can be found in the work
of Latouche and Ramaswami~\cite{LatoucheR87} and of Neuts~\cite{neuts1995matrix}, who pioneered the application of matrix
geometric methods for the solution of infinite M/G/1 queues.  
The literature in the field is rich of optimised methods derived using
specific properties of the transition matrix. The most remarkable
achievements are summarised in a comprehensive
book~\cite{bini2005numerical}.  In this paper, we consider numeric
methods as a basis for comparison but our main focus is on analytical
closed form solutions.

Mills and Anderson~\cite{mil11-rtcsa} have recently
considered the problem of stochastic analysis for resource
reservations on multiprocessor systems.  The authors main focus is on
the computation of tardiness and response time bounds for the average
case. The authors also offer a very conservative result on the
probabilistic deadlines, which is applicable
only if deadlines much larger than the period are considered.

A customary assumption made in the literature on queueing networks is
that inter--arrival times and service times are i.i.d. processes. In this
paper, we stick to the same assumption. Different authors have
recently questioned on the applicability of the i.i.d. assumption in
the area of real--time applications~\cite{Santos2215}. Remarkable is
the so called notion of probabilistic worst case execution
time~\cite{bernat2005probabilistic}, which essentially corresponds
to associating a worst case to several execution scenarios that take
place within a given probability.  A possible evolution of this
concept could lead to finding an i.i.d. overapproximation for a computation
process that is not i.i.d. A similar idea underpins a recent work by
Liu et al.~\cite{Liu14}, where the authors tackle the correlation
problem decomposing the process into a deterministic and an
i.i.d. component. In a similar context our results could be used to
study the evolution of the system under the action of the
i.i.d. component or of the i.i.d. overapproximation of the process.

A complementary issue to our work is how to derive statistically
sound estimates for the probability distribution of the computation
time. A useful inspiration could come from the application of the
Extreme Value Theory~\cite{6257562}, but the matter is reserved
for future investigations.

The results shown in this paper take to its natural completion a line
of work started a few years ago that has produced a number of
intermediate results. The relation with our prior achievements is
detailed in Section~\ref{sec:discussion}.

\section{Problem Description}
\label{sec:problem}
\subsection{Task Model}
We consider a set of real--time tasks $\left\{ \tau_{i} \right\} $
sharing a {\em processing unit} (CPU).  A real--time task $\tau_{i}$
consists of a stream of jobs $J_{i, k}$. Each job $J_{i, k}$ arrives
(becomes eligible for execution) at time $r_{i, k}$, and finishes at time $f_{i,
  k}$ after executing for a time $c_{i, k}$.  We restrict to periodic
tasks, meaning that two adjacent arrivals are spaced out by a fixed
amount of time $T_i$: $r_{i,\,k+1}=r_{i\,k}+T_i$.

The computation time of each job $c_{i, k}$ is assumed to be an i.i.d.
stochastic process $\mathcal{U}_i$. For each $k$ the computation time
is a random variable described by the Probability Mass Function (PMF)
$U_i(c) = \Prob{c_{i, k}=c}$.

Job $J_{i,\,k}$ is associated with a deadline $d_{i, k} = r_{i,\, k} +
D_i$ (where $D_i$ is said relative deadline), that is respected if
$f_{i,\, k} \leq d_{i,\, k}$, and is missed if $f_{i,\, k} > d_{i,\,
  k}$.  In this work, {\em probabilistic deadlines}~\cite{Abe98-3} are
used instead of traditional hard deadlines $d_{i, k}$.  A
probabilistic deadline $(D_i,\,p_i)$ is respected if $\Prob{f_{i, \,k}
  > r_{i,\, k} +D_i} \leq p_i$. If $p_i=0$ the deadline is hard.

\subsection{The scheduling algorithm}

As multiple real--time tasks may be concurrently active, we use a 
RR scheduler. Each
task $\tau_{i}$ is associated with a reservation $(Q^s_{i},T^s_{i})$,
meaning that $\tau_{i}$ is allowed to execute for $Q^s_{i}$
(\emph{budget}) time units in every interval of length $T_{i}^s$
(\emph{reservation period}).  The fraction of CPU allocated to the task
is said bandwidth $B_i$ and
is defined as $B_{i}=Q^s_{i}/T_{i}^s$.  The particular
implementation of the RR approach that we consider
is the Constant Bandwidth Server (\cbs)~\cite{Abe98-2}. In the
\cbs, reservations are implemented by means of an Earliest
Deadline First (EDF) scheduler.  The EDF schedules tasks
$\{\tau_{i}\}$ based on their \emph{scheduling deadlines} $d^s_{i,
  k}$, which are dynamically managed by the \cbs\/ algorithm.
When a new job $J_{i,k}$ arrives, the server checks whether it can be
scheduled using the last assigned scheduling deadline $d_{i,
  k-1}^{s}$. In the affirmative case, the scheduling deadline of the
job is initially set to current deadline $d_{i, k}^{s} = d_{i,
  k-1}^{s}$. Otherwise, the initial deadline $d^{s}_{i,k}$ is set
equal to $r_{i,k}+T^s_{i}$.  Every time the job executes for $Q^s_{i}$
time units (i.e., its budget is depleted), its scheduling deadline is
postponed by $T_i^s$: $d^{s}_{i,k} = d^{s}_{i,k} + T_i^s$.  This way,
the task is prevented from executing for more than $Q^s_{i}$ units
with the same deadline.  As a consequence, each task is reserved an
amount of computation time $Q^s_i$ in each server period $T_i^s$
regardless of the behaviour of the other tasks. This property is
called \emph{temporal isolation} and it holds as long as the system
satisfies the following \emph{schedulability condition}:
\begin{equation}
\sum_{i}B_{i} = \sum_i \frac{Q^s_i}{T^s_i}\leq 1. \label{eq:consistency}
\end{equation}

The scheduling deadline $d^s_{i,\, k}$ has, in general, nothing to do
with the deadline $d_{i,\, k}$ of the job: it is simply instrumental
to the implementation of the \cbs\/ (see~\cite{Abe98-2} for more details).

\subsection{Problem Statement}
In view of the temporal isolation property, each task is guaranteed a
minimum share of the processor $Q^s_i/T^s_i$ independently of the
behaviour of the other tasks. As a consequence, it is possible to
carry out a conservative analysis leading to the computation of a
lower bound of the probability of respecting a deadline assuming that
the task always receives this minimum (as long as
Condition~\eqref{eq:consistency} is respected).  The advantage is that
the behaviour of each task can be studied in isolation. Therefore, we
can remove the subscript $i$ meaning that the analysis refers to one
specific task.

In this setting, our problem is formulated as follows.
\begin{problem}
  Given a periodic real--time task with a stochastic computation time
  characterised by a PMF $U(c)$, find conditions on the
  reservation parameters $(Q^s, T^s)$ such that the task respects
  the probabilistic deadline $(D,\,p)$.
\end{problem}
A few remarks are in order. First of all, we look for analytical
conditions, which can be inverted and offer easy solution for the
problem of system design. Second, in order to be safely utilisable,
such conditions have to be \emph{sufficient} (although necessity is
certainly a desirable additional requirement).

\section{Stochastic Model}
 \label{sec:model}
 In this section, we first recall some basic definitions on Markov
 chains and in particular on QBDP.
  Then, we show how a task scheduled by a resource reservation
 is conveniently modelled as a QBDP (Theorem~\ref{th:QBDP}). Finally, 
we show how to derive a conservative approximation of this model, which
has a parametric accuracy and which retains the structure of a QBDP.

\subsection{Background on Markov Chains}
A \emph{Discrete--Time Markov Process} (DTMP) $\{X_n\}$ is a discrete--time
stochastic process such that its future development only depends on the current
state and not on the past history. This can be stated in formal terms
on the conditional PMF: $\Prob{X_n = x_n | X_1 = x_1, X_2 = x_2,
  \ldots, X_{n-1}=x_{n-1}} = \Prob{X_n = x_n | X_{n-1}=x_{n-1}}$.  A
DTMP defined over a discrete state space is said Discrete--Time Markov
chain (DTMC). Given a DTMC, let $\pi^{(j)}_n$ represent the
probability $\pi^{(j)}(n)=\Prob{X_n=j}$, $\pi_n$ be the vector $\pi_n
= [\pi_n^{(0)},\,\pi_n^{(1)},\ldots]$, $P=[p_{i,j}]$ be a matrix whose
generic element $p_{i,j}$ is given by the conditional probability
$p_{i,j}=\Prob{X_n=j|X_{n-1}=i}$.  Starting from an initial
probability distribution $\pi_0$, the application of the Bayes theorem
and of the properties of the Markov Processes allow us to express the
evolution of the distribution by the matrix equation
$\pi_{n+1}=\pi_{n} P$. The matrix $P$ is said probability transition
matrix. An \emph{equilibrium point} for this dynamic equation is a
vector $\tilde{\pi}$ such that $\tilde{\pi}=\tilde{\pi} P$.  

Consider a state $i$ of a DTMC. Let the random variable $\mathcal{T}_i
= \min \{n>1 \text{ s.t. } X_n = i | X_0 = i\}$ denote the first
return time to state $i$. The state $i$ is transient if
$\Prob{\mathcal{T}_i < \infty } < 1$, i.e., if there is some
probability that starting from $i$ the state will never return to
$i$. The state $i$ is {\em transient} if it is not recurrent. The
\emph{period} $d_i$ of a recurrent state $i$ is defined as the
greatest common divider of the set of all numbers, $n$, for which
$\Prob{X_{m}=i \wedge X_{m+n}=i} > 0, \forall m$.  A state is said
\emph{aperiodic} if its period $d_i = 1$. A DTMC is said aperiodic, if
all of its states are aperiodic.

The mean recurrence time of a state $i$
is the expected value of $\mathcal{T}_i$: $M_i = \Exp{\mathcal{T}_i}$.
The state $i$ is positive recurrent if $M_i$ is finite, and the DTMC
is positive recurrent if all its states are positive recurrent.

A DTMC is said \emph{irreducible}, if every state can be reached from
any other state in a finite number of steps. It can be shown that in an
irreducible DTMC all states are of the same type. So, if one
state is aperiodic, so is the DTMC.

A very important property of irreducible and positive recurrent DTMC
is \emph{the existence of a single equilibrium} $\tilde{\pi} = \tilde{\pi} P$
where the limiting distributions $\lim_{n \rightarrow \infty} \pi_n$
converge starting from any initial probability distribution $\pi_0$.
This equilibrium is called \emph{steady state distribution}.

A DTMC is called a 
Quasi--Birth--Death Process (QBDP)
 if its
probability transition matrix $P$ has the following block structure:
{\small
\begin{equation}
P=\begin{bmatrix}
C & A_0 & 0 & 0 &0 &\cdots\\
A_{2} &A_1 & A_0 & 0 &0 &\cdots\\
0 &A_{2} &A_1 & A_0 &0 &\cdots\\
0 &0 &A_{2} &A_1 & A_0 &\cdots\\
\cdots &\cdots & \cdots & \cdots & \cdots\\ 
\end{bmatrix}
\label{eq:tridiagonal}
\end{equation}
}
When the matrices are scalars, this structure reduces to the standard 
Birth--Death Process (BDP).

\subsection{A resource reservation as a Markov Chain}
\label{sec:full_model}
We will denote by $F_U(c) = \sum_{h=c_{min}}^c U(h)$
the Cumulative Distribution Function (CDF) of the execution time.  For simplicity, we will assume that the server period $T^s$
is chosen as an integer sub--multiple of the activation period $T$: $T
= NT^s$. Other choices are possible but make little
practical sense.

Let $d^s_k$ denote the latest scheduling deadline used for job $J_k$
and introduce the symbol $\delta_k = d^s_k-r_k$.  The latest
scheduling deadline $d^s_k$ is an upper bound for the finishing time
of the job (if Equation~\eqref{eq:consistency} is respected, then $f_k
\leq d^s_k$). Hence, $\delta_k$ is an upper bound for the job response
time.
\begin{example}
  Consider the schedule in Figure~\ref{fig:example}. The schedule in
  the figure considers two adjacent jobs starting at $r_k$ and
  $r_{k+1}$ and the reservation period is chosen as one third of the
  task period. Job $J_k$, in this case finishes beyond the deadline
  (which in our periodic model is $r_{k+1}$). More precisely, the last
  reservation period that it uses (in which its finishing time lies)
  is upper--limited by the scheduling deadline $d^s_k$.
\end{example}

The quantity $\delta_k$ takes on values in a
discrete set: the integer multiples of $T^s$
 and the
probability $p$ of meeting the deadline is lower bounded by
$\Prob{\delta_k \leq D}$.
\begin{figure}[t]
\centerline{\includegraphics[width=0.9\columnwidth]{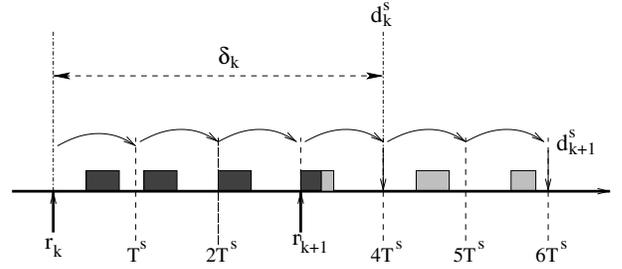}}
\caption{Example schedule of a task by a \cbs. The two colours denote
  different jobs.}
\label{fig:example}
\end{figure}

The evolution of $\delta_k$ is described as follows~\cite{Abe01-1}:
\begin{eqnarray}
v_0 & = & c_0 \nonumber \\
v_{k + 1} & = & \max \{0, v_k - N Q^s \} + c_{k + 1} \nonumber \\
\delta_k  & = & \left \lceil \frac{v_k}{Q^s} \right \rceil T^s
\label{eq:vprob}
\end{eqnarray}
The variable $v_k$ cannot be measured directly and it represents the
amount of backlogged execution time that has to be served by the \cbs\/
scheduler when a new job arrives.

Since the process $\mathcal{U}$ modelling the sequence $c_k$ of the
computation time is assumed a discrete valued and i.i.d. random
process, the model in Equation~\eqref{eq:vprob} represents a
Discrete--Time Markov Chain (DTMC) that we define $\mathcal{M}_0$,
where the states are determined by the possible values of $v_k$ and
the transition probabilities by the PMF of the computation time
$U(c)$.

This model permits a fine--grained modelling of the behaviour of the
reservation, which can be difficult to treat.  One possible
simplification is to collapse into a single state all
the states for which $\delta_k\leq D = N T^s$, which correspond to the
values of $v_k$ such that $v_k \leq NQ^s$. 
In the modified DTMC $\mathcal{M}$, the state $S$ is defined as
\[
S= \begin{cases} 
  0 & \text{if } v_k \leq NQ^s\\
  i & \text{if } v_k = NQ^s+i
\end{cases}.
\]
 By using Equation~\eqref{eq:vprob},
the transition probabilities for this DTMC can be written as follows:
{\small
\[
\begin{array}{l}
p_{i,j}
= \begin{cases}
\Prob{v_{k+1} \leq N Q^s| v_k=i+NQ^s},\,\text{if } j = 0\\
\Prob{v_{k+1}=j+NQ^s | v_k \leq N Q^s},\,\text{if } i = 0, j \neq 0\\
\Prob{v_{k+1} = N Q^s+j| v_k=i+NQ^s},\, \text{if } i \neq 0, j \neq 0\\

\end{cases} \\
= \begin{cases}
\Prob{c_k\leq NQ^s-i}=F_U(NQ^s-i),\,\text{if } j = 0\\
\Prob{c_k=j+NQ^s}=U(j+NQ^s),\,\text{if } i = 0, j \neq 0\\
\Prob{c_k = NQ^s+j-i}=U(j-i+NQ^s),\, \text{if } i \neq 0, j \neq 0.
\end{cases}.
\end{array}
\]
}
Let $\tilde{\pi}_k$ be the (infinite) vector where the $i^{th}$
element represent the probability associated with the $i^{th}$ state
of the DTMC $\mathcal{M}$ after $k$ step of evolution starting from an
initial probability vector $\tilde{\pi}_0$. The recursive equation for the
evolution of $\tilde{\pi}_k$ is $\tilde{\pi}_{k+1}=\tilde{\pi}_k P$.
The objective of our analysis can now be stated as 
\emph{the computation of a lower bound for the first
  element of the steady state probability vector} $\tilde{\pi} =
\lim_{k \rightarrow \infty} \pi_k$.  As long as we are not interested in the
distribution of $\delta_k$ inside the region $\delta_k \leq N Q^s$,
collapsing into one state all the values of $v_k$ smaller than $NQ^s$
does not introduce any error because such states do not have influence
on the next state ($\max \{0, v_k - N Q^s \}=0$ in
Equation~\eqref{eq:vprob}).

The probability matrix $P$ resulting from the computation above has the structure shown in
Figure~\ref{fig:structure}, where 
\[
\begin{array}{l}
a_{H+h} = p_{i,\,i+h} = U(h+NQ^s)\\
b_{H-i} = p_{i,\,0} = F_U(NQ^s-i),
\end{array}
\]
and $H$ is the minimum integer such that $U(NQ^s+h)=0$ for all $h <
H$. This structure is recursive: from row $H$ onward, each row is
obtained by shifting the previous one to the right and inserting a $0$
in the first position.  Furthermore, the first element greater
  than zero of such recursive rows is dubbed $a_0$, while the last
  with $a_n$: $n = \max\{i | a_i > 0\}$.
\begin{figure}[t]
  \[
{\small
\begin{bmatrix}
  b_H & a_{H+1} &\ldots  & a_n & 0 &\ldots \\
  b_{H-1} & a_{H} & a_{H+1} &\ldots  &a_n& \ldots \\
  b_{H-2} & a_{H-1} & a_{H} & a_{H+1} &\ldots  &a_n& \ldots \\
  \ldots & \ldots &\ldots  &\ldots& \ldots &\ldots \\
  b_1 & a_2 &\ldots& a_{H} & a_{H+1} &\ldots  &a_n& \ldots \\
  a_0 & a_1 & a_2 &\ldots & a_{H}& a_{H+1} &\ldots  & \ldots \\
  0 & a_0 & a_1 &a_{h+H-4} &\ldots& a_{H} & a_{H+1} & \ldots \\
  0 & 0 & a_0 &a_1 & a_2 &\ldots& a_{H} & \ldots \\
  \vdots & \vdots & \vdots & \vdots &  \ddots &\ddots &\ddots\\
\end{bmatrix}},
\]
\caption{Structure of the transition matrix $P$}
\label{fig:structure}
\end{figure}
We now introduce a useful notation for sub--matrices.
\begin{definition}
  Let $P=(p_{i,\,j})$ be a matrix whose elements are $p_{i,j}$.  Let
  $\alpha=\left\{i_i,\,i_2,\ldots, \,i_n\right\}$
  $\beta=\left\{j_i,\,j_2,\ldots, \,j_m\right\}$ two ordered set of
  indexes. The sub--matrix $P_{[\alpha,\,\beta]}$ is a matrix whose
  elements are $p_{i_{h}, j_{t}}$ for all $h \in [1,\,n]$ $t \in
  [1,\,m]$.  Likewise, if $\pi$ is a vector, we denote $\pi_{[\alpha]}$
  the sub--vector whose elements are $\pi_{i_{h}}$ for all
  $h \in [1,\,n]$.
\label{def:submatrix}
\end{definition}

From the properties of our transition matrix we can prove the 
 following result~\cite{etfa2012}.
\begin{theorem}
\label{th:QBDP}
  Let $H$ be the minimum integer such that
  $U(NQ^s+h)=0$ for all $h < H$.  Let $F$ be defined as
  $\max\left\{n-H,H\right\}$. Define \al{i}{F} the set
  $\{i,\ldots,\,i+F-1\}$ and $\be{j}{F}$ the set
  $\{j,\ldots,\,j+F-1\}$.  The transition matrix $P$ is
  block--tri--diagonal with the structure in Equation~\ref{eq:tridiagonal},
where $A_0=P_{[\al{F}{F}, \be{0}{H}]}$, $A_2=P_{[\al{0}{F}, \be{F}{F}]}$, 
$A_1=P_{[\al{F}{F}, \be{F}{F}]}$, $C=P_{[\al{0}{F}, \be{0}{F}]}$, are square matrices of order $H.$ This qualifies the process
as a QBDP.
\end{theorem}
The structure of the QBDP exposed in Theorem~\ref{th:QBDP} enables the
application of efficient numeric solutions for the steady state
probability~\cite{bini2005numerical}, as discussed in
Section~\ref{sec:experiments}.

\subsection{A conservative approximation}
\label{sec:simplification}
In order to make the model tractable from the numeric point
of view, it is useful to
introduce a conservative approximation. 
The notion of conservative approximation that we shall adopt here
relies on the concept of {\em first order stochastic dominance}
(defining an order relation between probability distributions):
\begin{definition}
Given two random variables $X$ and $Y$, with CDFs $F_x(x)$ and $F_y(y)$,
$X$ has a first order stochastic dominance over $Y$ ($X \succeq Y$)
iff $\forall x$ $F_x(x)\leq F_y(x)$.
\label{def:Succ}
\end{definition}
Based on this definition, a stochastic real--time task can be seen as a conservative
approximation of another one if its probabilistic deadlines are
stochastically dominated by the probabilistic deadlines of the
original task: considering $\delta_k$ in Equation~\eqref{eq:vprob},
this plainly means that in the modified system the low values of the
$\delta_k$ will have a greater probability and so will be the
probability of the first element of the probability vector (associated
with the deadline satisfaction).

As shown by Diaz et al.~\cite{Dia04}, if $\mathcal{U}'$ stochastically dominates $\mathcal{U}$,
then a system having the execution times distributed according to
$\mathcal{U}'$ is a conservative approximation of the original system (with 
the execution times distributed according to $\mathcal{U}$).

A simple way to build $\mathcal{U}'$ to obtain such a conservative
approximation is to replace $c_k$
with a new variable $c'_k$ whose distribution is given by:
\begin{equation}
   U_\Delta (c') = \begin{cases} 
    0 & \text{if } c' \mod \Delta \neq 0\\
    \sum_{c=(k-1) \Delta + 1}^{k \Delta}U(c') &  \text{otherwise},
\end{cases}
\label{eq:uh}
 \end{equation}
where $\Delta$ is a scaling factor chosen as an integer sub--multiple of $Q^s$.
The transition matrix of the new DTMC has again the
structure in Fig.~\ref{fig:structure}, where the different 
elements of the matrix are functions of the parameter $\Delta$. 
Large values of $\Delta$ correspond to a smaller size for matrices
$A_2$, $A_1$, $A_0$ in Equation~\ref{eq:tridiagonal}. This reduces the
time required for the computation of the steady state probability
paying the price of a coarser approximation for the computed
probability.

\section{An analytical bound}
\label{sec:analytical}

This section presents an analytic solution for a QBDP described by the
transition matrix reported in Fig.~\ref{fig:structure}.
In the discussion, we assume that the conservative approximation discussed
in Section~\ref{sec:simplification} for some $\Delta$. 

The first key result of the Section is Theorem~\ref{th:Final}, which
shows a general expression for the steady state
probability of respecting the deadline. After introducing an
additional simplification in the model, this leads to the analytical
bound in Theorem~\ref{th:FinalClosedForm} and in
Corollary~\ref{cor:FinalClosedForm}, which represent the core theoretical 
results of
the paper.

\subsection{A solution for generic QBDP processes}
\label{sec:GeneralSolution}

  Before
going into the theoretic details, let us define the following function
$\gamma: \natu\times\real \rightarrow \real$ as
\[
\gm{k,l} = \sum_{j=0}^{k} \alpha_{j} l^{k-j} ,
\]
where $\alpha_j = a_j/a_0$. Using this function and the structure of
the QBDP, it is possible to write the equation expressing the
steady state equilibrium  $\tilde\pi_k = \tilde\pi_k P$, (where $\tilde \pi_k = \left[\tilde\pi_k^{(0)}, \tilde\pi_k^{(1)}, \ldots\right]$) by 
expressing the probabilities $\tilde\pi^{(i)}_k$, $i > H$, at time $k$
as a function of $\tilde\pi^{(j)}_k$, $0\leq j\leq H$, in the
following way:
  {\footnotesize
\begin{equation}
  \label{eq:RecursivePi}
\begin{aligned}
  \tilde{\pi}^{(H)}_k & = \sum_{j=H+1}^{n} \alpha_j \tilde{\pi}^{(0)}_k - \sum_{j=1}^{H-1} \gm{j,1} \tilde{\pi}^{(H-j)}_k, \\
  \tilde{\pi}^{(H+l)}_k & = \left (\gm{H-1,1} + \sum_{j=H+1}^{n} \alpha_j \right
  ) \tilde{\pi}^{(l)}_k -\!\!\!\!\!\! \sum_{\substack{ j=1\\ j\neq H}}^{min(n,l+H)}\!\!\!\!\!\!
  \alpha_{j} \tilde{\pi}^{(l+H-j)}_k,
\end{aligned}
\end{equation}}
holding for $\forall l > 1$.

The steady state solution for generic $n > H > 0$ is given by the
following theorem:
\begin{theorem}
\label{th:Final}
Consider a QBDP described by the transition probability matrix $P$
given in Fig.~\ref{fig:structure}, in which both $a_0$ and
$a_{n}$ differ from zero.

Assume that the matrix {\tiny
\begin{equation}
\label{eq:MatW}
W =
\begin{bmatrix}
 0 & 1 & 0 & \ldots & 0 & 0 & \ldots & 0 \\
 0 & 0 & 1 & \ldots & 0 & 0 & \ldots & 0 \\
 0 & 0 & 0 & \ddots & 0 & 0 & \ldots & 0 \\
 \vdots & \vdots & \vdots & \ddots & \ddots & \ddots & \ddots & \vdots
 \\
 0 & 0 & 0 & \ldots & 0 & 0 & \ldots & 1 \\
 -\alpha_n & -\alpha_{n-1} & -\alpha_{n-2} & \ldots & w & -\alpha_{H-1}
 &  \ldots & -\alpha_0 \\
 \end{bmatrix}
\end{equation} }
where $w = \gm{H-1,1} + \sum_{j=H+1}^{n} \alpha_j$, has distinct
eigenvalues. Let $\pi^{(j)} =\lim_{k\rightarrow +\infty} \tilde{\pi}^{(j)}_k$ be the steady state
distribution of the state. 
One of the two following cases apply:
\\
I)
if $\sum_{j=0}^{H-1} \gm{j,1} \leq \sum_{j=H+1}^{n}
(j-H) \alpha_j$ then the limiting distribution is given by:
\begin{equation}
  \label{eq:FinalResult1}
    \tilde{\pi}^{(j)} =\lim_{k\rightarrow +\infty} \tilde{\pi}^{(j)}_k = 0,\,\,\, \forall
  j,
\end{equation}
II) if $\sum_{j=0}^{H-1} \gm{j,1} > \sum_{j=H+1}^{n}
(j-H) \alpha_j$ then: 
\begin{equation}
  \label{eq:FinalResult2}
     \tilde{\pi}^{(0)} = \prod_{\beta\in \mathcal{B}_s}(1 - \beta) .
\end{equation}
In the second case, $\mathcal{B}_s$ is the set of stable eigenvalues of $W$ (in this context
an eigenvalue $\beta$ is said stable if $|\beta|<1$), and
the terms $\tilde{\pi}^{(j)}$ with $0 < j < H$ are known
linear functions of $\tilde{\pi}^{(0)}$, while the terms
$\tilde{\pi}^{(j)}$ with $j\geq H$ are given
by~\eqref{eq:RecursivePi}.
\end{theorem}
Before showing the proof, we make two important remarks.
\begin{remark}
The assumption on the eigenvalues of the matrix $W$ is
merely technical (it simplifies the proof of the result) and
it is not restrictive.  In all our examples (both synthetically
generated and using data from real applications), it is
respected. Artificial examples that violate it could probably be
constructed but they are not relevant in practice.
\end{remark}
\begin{remark}
\label{rem:companion}
As well as paving the way for Theorem~\ref{th:FinalClosedForm},
Theorem~\ref{th:Final} contains an implicit numeric algorithm for the
computation of $\tilde{\pi}^{(0)}$, based on the computation of the
eigenvalues of the matrix $W$. Since the latter is in companion form,
in the following we refer to this algorithm as {\tt companion}.
\end{remark}

\subsection{Proof of Theorem \ref{th:Final}}  
This section is devoted to the proof of the fundamental
Theorem~\ref{th:Final}, which will require several definitions and
auxiliary results.  
The section can be skipped over  if the reader 
is only interested in the applications of the Theorem.

The rationale behind the proof is the
following. First, the equilibrium point of the QBDP is expressed as an
iterative system. The evolution in the iteration step represents the
connection between the probabilities of the different states.  Using
this representation and some property of convergence of the Markov
chain, we can express all the steady--state probabilities as a
function of $\tilde{\pi}^{(0)}$, which can eventually be found
as a solution of a linear system of equations.

 We start noticing that having $a_0$ and $a_{n}$ different from zero
 implies that the Markov chain of the QBDP is irreducible and
 aperiodic.  Therefore, it is guaranteed that the probability of the
 different states converge to a value~\cite{CassandrasL06}. Notice,
 however, that this does not necessarily imply the existence of a
 steady--state distribution (the distribution could shift toward
 increasing values of the state without ever reaching the
 equilibrium, with the probability of each state going to $0$).

\paragraph{The case of Positive Recurrent QBDP}
If the QBDP is positive recurrent, it admits indeed a unique steady state
  distribution.  The first step of the proof is then to introduce the
  following vector: $\Pi_j =
  [\tilde{\pi}^{(j)},\dots,\tilde{\pi}^{(j+n-1)}]^T$, whose dimension is
  equal to $n$. It is possible to exploit~\eqref{eq:RecursivePi}
  and~\eqref{eq:MatW} to derive the equilibrium of the QBDP 
  by the following iterative equation for the vector $\Pi_j$:
  {\small
  \[
\begin{aligned}
\Pi_1 & = \begin{bmatrix}\tilde{\pi}^{(1)}\\\tilde{\pi}^{(2)}\\\vdots\\\tilde{\pi}^{(n)}
\end{bmatrix} = W\Pi_0 \Rightarrow \Pi_j = \begin{bmatrix}\tilde{\pi}^{(j)}\\\tilde{\pi}^{(j+1)}\\\vdots\\\tilde{\pi}^{(n-1+j)}
\end{bmatrix} = W^j \Pi_0 .
\end{aligned}
\]}
Using this notation the normalisation constraint $\sum_{h=0}^\infty \tilde{\pi}^{(h)} = 1$ can be expressed as
{\small
\begin{equation}
  \label{eq:EqDyn}
  \sum_{h=0}^\infty \tilde{\pi}^{(h)}=
  \begin{bmatrix}
  1 & 0 & 0 & \dots & 0
  \end{bmatrix} \sum_{i=0}^{+\infty} \Pi_i = 1.
\end{equation}}

The characteristic polynomial of the lower--left {\em companion form}
matrix $W$ reported in~\eqref{eq:MatW} is simply given by
\begin{equation}
  \label{eq:PolyAlpha}
  P(\lambda) = \lambda^n - \left (\gm{H-1,1} + \sum_{j=H+1}^{n} \alpha_j \right
  )\lambda^{n-H} +\sum_{\substack{j=1\\ j\neq
        H}}^{n} \alpha_{j} \lambda^{n-j} ,
\end{equation}
from which it is trivially derived that the matrix $W$ has one simple
eigenvalue in $\beta_1 = 1$ and additional $n-1$ eigenvalues
$\beta_i$. Therefore
\begin{equation}
  \label{eq:PolyBeta}
  P(\lambda) = (\lambda - 1)\prod_{i=2}^{n}(\lambda -
  \beta_i) .
\end{equation}
Since each $\beta_i$ verifies $P(\beta_i) = 0$, the following relation
holds
 {\small
\begin{equation}
  \label{eq:RelBetaPoly}
   \begin{aligned}
& \beta_i^n - \left (\gm{H-1,1} + \sum_{j=H+1}^{n} \alpha_j \right
  )\beta_i^{n-H} +\sum_{\substack{j=1\\
      j\neq H}}^{n} \alpha_{j}
  \beta_i^{n-j} = 0 \Rightarrow \\ 
    & \gm{H-1,1} + \sum_{j=H+1}^{n} \alpha_j = \beta_i
    \gm{H-1,\beta_i} + \frac{\sum_{j=H+1}^{n} \alpha_j \beta_i^{n-j}}{\beta_i^{n-H}} .
  \end{aligned}
\end{equation}}

Since all the eigenvalues are assumed simple, we can use of the {\em
  spectral decomposition} of the matrix $W$: $W = \sum_{i=0}^{n-1}
\beta_i G_i$, where the {\em spectral projectors} $G_i$ are given by
$
G_i = \frac{V_{i} L_{i}}{L_{i}V_{i}} = N_i V_{i} L_{i} ,
$
and $L_i$ and $V_i$ are respectively the left and right eigenvectors
associated with the $i$--th eigenvalue $\beta_i$. $N_i$ is the
normalisation constant needed to satisfy the spectral projectors basic
properties, i.e., $G_iG_j = 0$ for $i\neq j$ and $G_iG_i = G_i$.  As a
consequence,
$
\Pi_1 = W \Pi_0 = \sum_{i=1}^{n} \beta_i G_i \Pi_0 ,
$
and, in general,
  {\small
\begin{equation}
  \label{eq:SumGi}
\Pi_j = W^j \Pi_0 = \sum_{i=1}^{n} \beta_i^j G_i \Pi_0 = \sum_{i=1}^{n} \beta_i^j N_i V_{i} L_{i} \Pi_0 .
\end{equation}}
Therefore, by combining~\eqref{eq:SumGi} and~\eqref{eq:EqDyn}, one gets:
  {\small
\begin{equation}
  \label{eq:NormProb}
  \sum_{i=1}^{n}\sum_{k=0}^{+\infty} \beta_i^k v_i^{(0)} N_i L_{i}
  \Pi_0 = 1,
\end{equation}}
where $v_i^{(0)}$ is the first element of the right eigenvector.
Given the expression of the matrix $W$, the left $L_i$ and right $V_i$
can be easily found as a function of $\beta_i$.
From the expression of the eigenvectors, it follows immediately that
  {\small
\begin{equation}
  \label{eq:NFactor}
\begin{aligned} 
  N_i & = \frac{1}{L_i V_i} = \frac{\displaystyle \beta_i^{n}}{\displaystyle \sum_{j=0}^{H-1}
    \gm{j,\beta_i} \beta_i^{n-j} -\sum_{j=H+1}^{n} (j-H) \alpha_j \beta_i^{n-j}}
  .
\end{aligned}
\end{equation}}
We now state some auxiliary propositions on vector $\Pi_0$.

\begin{proposition}
\label{prop:EigenOrtho}
The product between the left eigenvector $L_i$ and the 
initial vector of the iteration $\Pi_0$ is given by
{\small
\[
L_i \Pi_0 = \beta_i^{n-H-1} \left ( \beta_i - 1 \right ) \left (
  \sum_{k=0}^{H-1} \sum_{j=k}^{H-1} \gm{H-1-j,\beta_i} \tilde{\pi}^{(k)} \right ).
  \]}
\end{proposition}

\begin{IEEEproof}
  The proof of the proposition follows by first computing the explicit
  computation of the product $L_i \Pi_0$, in which each term is
  substituted with the recursive Equations~\eqref{eq:RecursivePi} and
  the constraint given in~\eqref{eq:RelBetaPoly}, and then noticing
  that 
    {\small
  \[
 \beta_i^n - 1 = (\beta_i - 1) \sum_{j=0}^{n-1}
    \beta_i^j . 
  \]}
  See~\cite{TR} for more details.

\end{IEEEproof}

\begin{proposition}
\label{prop:beta1}
  The initial vector $\Pi_0$ is orthogonal to the left eigenvector
  associated to $\beta_1 = 1$.
\end{proposition}
\begin{IEEEproof}
  The proof follows from Proposition~\ref{prop:EigenOrtho}.
\end{IEEEproof}

\begin{proposition}
\label{prop:UnstableEigen}
For any unstable eigenvalue $\beta_i$ (i.e., such that $|\beta_i| > 1$) of $W$ it holds that $L_i \Pi_0 = 0$.
\end{proposition}
\begin{proof}
  If the QBDP has an equilibrium then~\eqref{eq:NormProb} holds true.
  The unitary eigenvalue $\beta_1 = 1$ does not play any role in the
  summation of~\eqref{eq:NormProb} in view of
  Proposition~\ref{prop:beta1}.  Next, suppose that there exists one
  or more $|\beta_i| > 1$.  From Equation~\eqref{eq:NormProb} it
  follows that it may be $L_i \Pi_0 = 0$, $N_i = 0$ or $\Pi_0 = 0$.
  Since the normalisation factor cannot be null, let us first consider
  $\Pi_0 = 0$.  Using~\eqref{eq:SumGi} it follows that $\Pi_0 = 0
  \Rightarrow \Pi_j = 0$, $\forall j$.  Therefore,
    {\small
  \[
  \tilde{\pi}^{(j)} =\lim_{k\rightarrow +\infty} \tilde{\pi}^{(j)}(k) = 0,\,\,\,
  \forall j ,
  \]}
  and, since the Markov chain is irreducible and aperiodic, the QBDP
  does not have a unique stationary distribution~\cite{CassandrasL06},
  which contradicts the hypothesis.

  It then follows that for any unstable eigenvalue $L_i \Pi_0 = 0$.
\end{proof}

From Rouche's theorem~\cite{Lloyd79} we have that the number of
eigenvalues $\beta_i$ such that $|\beta_i| \geq 1$ of the matrix $W$
is exactly equal to $H$, where $H-1$ have $|\beta_i| > 1$.  The
consequences of Proposition~\ref{prop:UnstableEigen} are twofold.
First, it states that Proposition~\ref{prop:EigenOrtho} defines $H-1$
linear equations
 {\small
\begin{equation}
  \label{eq:RecEq2}
\sum_{k=0}^{H-1} \sum_{q_1=0}^{H-1-k} \gm{q_1, \beta_i} \tilde{\pi}^{(k)} = 0
, \forall \beta_i\in \mathcal{B}_s^\star,
  \end{equation}}
where $\mathcal{B}_s^\star$ is the set of $H-1$ unstable eigenvalues
except $\beta_1 = 1$ (the unstable eigenvalue $\beta_1$ does not play
any role by Proposition~\ref{prop:beta1}).  The $H$ unknown
probabilities $\tilde{\pi}^{(0)}$ to $\tilde{\pi}^{(H-1)}$
of~\eqref{eq:RecEq2} are also the unknowns of the recursion
formulae~\eqref{eq:RecursivePi}. The second consequence is that
  {\small
\begin{equation}
  \label{eq:NormProb2}
  \sum_{\beta_i\in\mathcal{B}_s}\frac{v_i^{(0)} N_i}{1 - \beta_i} L_{i}
  \Pi_0 = 1 ,
\end{equation}}
where $\mathcal{B}_s$ is the set of stable eigenvalues.  By
substituting in~\eqref{eq:NormProb2} the result given in
Proposition~\ref{prop:EigenOrtho} and the expression of the right eigenvector $L_i$, we get
 {\small
\begin{equation}
  \label{eq:NormProb3}
  - \sum_{\beta_i\in\mathcal{B}_s}\frac{N_i}{\beta_i^H}
  \sum_{k=0}^{H-1} \sum_{q_1=0}^{H-1-k} \gm{q_1,\beta_i} \tilde{\pi}^{(k)} = 1 .
\end{equation}
}
By means of Proposition~\ref{prop:UnstableEigen}, the summation can be
extended to the unstable eigenvalues, except for the first eigenvalue
$\beta_1 = 1$, which instead induces indefiniteness
of~\eqref{eq:NormProb3}.  The solution to~\eqref{eq:NormProb3} is
derived exploiting the spectral projectors property $\sum_{i=1}^n G_i
= I_{n}$.  Indeed, summing the elements in position $(n-H,n-j)$, for
$1\leq j\leq H-1$, we have for each $j$
{\small
\[
- \sum_{i=1}^{n} N_i v_i^{(n-H)} l_i^{(n-j)} = - \sum_{i=1}^{n}
\frac{N_i}{\beta_i^H} \gm{j,\beta_i} = 0,
\]}
and hence
{\small
\[
- \sum_{i=2}^{n} \frac{N_i}{\beta_i^H} \gm{j,\beta_i} = N_1 \gm{j, 1} ,
\]}
where $N_1$ is easily obtained by~\eqref{eq:NFactor} for $\beta_1 =
1$, i.e.,
{\small 
\[
N_1 = \frac{\displaystyle 1}{\displaystyle \sum_{j=0}^{H-1} \gm{j,1}
  -\sum_{j=H+1}^{n} (j-H) \alpha_j } = \frac{1}{D_1}.
\]}
Moreover, for the elements in position $(n-H+1,1)$, we get
{\small
\[
- \sum_{i=1}^{n} N_i v_i^{(n-H+1)} l_i^{(1)} = \sum_{i=1}^{n}
\frac{N_i}{\beta_i^{H-1}} \frac{\alpha_n}{\beta_i} = 0 \Rightarrow -
\sum_{i=2}^{n} \frac{N_i}{\beta_i^H} = N_1 .
\]}
Substituting these relations in~\eqref{eq:NormProb3} produces the
equation
  {\small
\begin{equation}
  \label{eq:RecEq1}
\sum_{k=0}^{H-1} \sum_{q_1=0}^{H-1-k} \gm{q_1, 1} \tilde{\pi}^{(k)} = D_1 ,
\end{equation}}
which, used in conjunction with the $H-1$ equations
of~\eqref{eq:RecEq2}, determines the set of unknown probabilities.

In order to have an analytic solution of this linear system of $H$
equations in $H$ unknowns, we start by collecting the probability with
the highest index, i.e.,
{\small
\[
\begin{aligned}
  \tilde{\pi}^{(H-1)} & + \sum_{k=0}^{H-2} \sum_{q_1=0}^{H-1-k} \gm{q_1, 1}
  \tilde{\pi}^{(k)} = D_1 \\
  \tilde{\pi}^{(H-1)} & + \sum_{k=0}^{H-2} \sum_{q_1=0}^{H-1-k} \gm{q_1,
    \beta_i} \tilde{\pi}^{(k)} = 0, \beta_i\in\mathcal{B}_s^\star,
\end{aligned}
\]}
from which it is possible to immediately have the solution
{\small
\[
\tilde{\pi}^{(H-1)} = - \sum_{k=0}^{H-2} \sum_{q_1=0}^{H-1-k} \gm{q_1,
  \beta_H}
\tilde{\pi}^{(k)}
\]}
and the $H-1$ new linear equations in $H-1$ unknowns
{\small
\[
\sum_{k=0}^{H-2} \sum_{q_1=0}^{H-1-k} \left (\gm{q_1, 1} - \gm{q_1,
    \beta_i} \right ) \tilde{\pi}^{(k)} = D_1, \beta_i\in\mathcal{B}_s^\star ,
\]}
that, by simple algebraic manipulations, leads to
{\small
\[
\sum_{k=0}^{H-2} \sum_{q_1=0}^{H-1-k} \sum_{q_2 = 0}^{q_1-1}
\gm{q_2, \beta_i} \tilde{\pi}^{(k)} = \frac{D_1}{1 - \beta_i},
\beta_i\in\mathcal{B}_s^\star.
\]}
From the new set of $H-1$ equations the element
$\tilde{\pi}^{(H-2)}$ can be collected, thus leading to a recursive
solution formula.
The recursion can be executed for $H$ steps until the following final
equation is obtained
  {\small
\begin{equation}
  \label{eq:Pi0Unst}
\tilde{\pi}^{(0)} = \frac{\displaystyle D_1}{\displaystyle
  \prod_{\beta_i\in\mathcal{B}_s^\star} (1 - \beta_i)} = \frac{\displaystyle \sum_{j=0}^{H-1} \gm{j,1}
  -\sum_{j=H+1}^{n} (j-H) \alpha_{j}}{\displaystyle
  \prod_{\beta_i\in\mathcal{B}_s^\star} (1 - \beta_i)} .
\end{equation}}

The result in~\eqref{eq:Pi0Unst} can be suitably rewritten in a more
useful way.  
To this end, we first rewrite the characteristic
polynomial~\eqref{eq:PolyBeta} as follows
  {\small
\begin{equation}
  \label{eq:Prop3Eq}
  P(\lambda) = (\lambda -
  1)\prod_{i=2}^{n-1}(\lambda - \beta_i) = \lambda^{n-1} + \sum_{j=1}^{n-1}
  \mathcal{S}_j(\beta) \lambda^{j-1} ,
\end{equation}}
where
 {\small
\begin{equation}
  \label{eq:Prop3Eq2}
  \begin{aligned}
    & \mathcal{S}_j(\beta) = (-1)^{n-j+1} \left (
      \sum_{J\in\mathcal{C}_{1}}\prod \beta_J +
      \sum_{J\in\mathcal{C}_{2}}\prod \beta_J \right ) ,
  \end{aligned}
\end{equation}}
and where $\mathcal{C}_{1}$ and $\mathcal{C}_{2}$ are proper sets of
indices coming from the explicit computation of the characteristic
polynomial. Since the product of all the eigenvalues, except for the
first one, is given by
{\small
\[
\prod_{i=2}^{n} (1 - \beta_i) = 1 + \sum_{j=1}^{n-1}
(-1)^{n-j}\sum_{J\in\mathcal{C}_{n-j}}\prod \beta_J = 1 +
\sum_{j=1}^{n-1} \mathcal{W}_j(\beta) ,
\]}
where, by means of~\eqref{eq:Prop3Eq2}, $\mathcal{W}_k(\beta) = -
\sum_{j=1}^{k} \mathcal{S}_j(\beta)$, one gets
{\small
\begin{equation}
  \label{eq:RecursiveProd}
\prod_{i=2}^{n} (1 - \beta_i) = 1 - \sum_{j=1}^{n-H} \sum_{k=1}^{j}
\mathcal{S}_k(\beta) - \sum_{j=n-H+1}^{n-1} \sum_{k=1}^{j}
\mathcal{S}_k(\beta) .
\end{equation}}
From~\eqref{eq:Prop3Eq} and~\eqref{eq:PolyAlpha},
$\mathcal{S}_k(\beta) = \alpha_{n-k+1}$, for $1\leq k\leq n$, and
$\mathcal{S}_k(\beta) = \gm{H-1,1} + \sum_{j=H+1}^{n} \alpha_j$, for
$k=n-H+1$.  Substituting these relations in the last two terms
of~\eqref{eq:RecursiveProd}, one gets
{\small
\[
\begin{aligned}
& - \sum_{j=1}^{n-H} \sum_{k=1}^{j} \mathcal{S}_k(\beta) =
- \sum_{j=H+1}^{n}(j-H)\alpha_j ,\\
& - \sum_{j=n-H+1}^{n-1} \sum_{k=1}^{j} \mathcal{S}_k(\beta) = (H-1)
\gm{H-1,1} - \sum_{j=1}^{H-1} (j-1) \alpha_j .
\end{aligned}
\]}
Since
{\small
\[
1 + (H-1) \gm{H-1,1} - \sum_{j=1}^{H-1} (j-1) \alpha_j =
\sum_{j=0}^{H-1} \gm{j,1} ,
\]}
Equation~\eqref{eq:RecursiveProd} is rewritten as
{\small
\begin{equation}
  \label{eq:EigProdSimply}
  \prod_{i=2}^{n} (1 - \beta_i) = \sum_{j=0}^{H-1} \gm{j,1} -
  \sum_{j=H+1}^{n}(j-H)\alpha_j = D_1 ,
\end{equation}}
that substituted in~\eqref{eq:Pi0Unst} finally yields
Equation~\eqref{eq:FinalResult2}.

At this point we have proved that \emph{if the QBDP has an
  equilibrium}, this is given by~\eqref{eq:FinalResult2}, by the recursive
solution of the linear system of equations~\eqref{eq:RecEq1}
and~\eqref{eq:RecEq2}, and by the recursion
formula~\eqref{eq:RecursivePi}.

\paragraph{The case of non--positive recurrent QBDP}
If the QBDP is not positive recurrent we can re--write matrix $P$ using
its block--tridiagonal representation in~\eqref{eq:tridiagonal}.
We can immediately apply the following theorems.

\begin{theorem}
  \label{th:StatDistr}
  \cite{CassandrasL06} An irreducible Markov chain has a stationary
  distribution if and only if all its states are positive recurrent.
\end{theorem}

\begin{definition}
  Assume $A = A_0 + A_1 + A_2$ is irreducible. Then, by the
  Perron--Frobenius Theorem, there exists a unique vector $\mu > 0$
  with ${\bf 1}^T \mu = 1$ and $A \mu = \mu$.  The vector $\mu$ is
  called the stationary probability vector of $A$, while ${\bf 1}$ is
  a column vector whose elements are all equal to one.
\end{definition}

\begin{theorem}
  \label{th:Transient}
  \cite{LatoucheR87} The QBDP is transient if ${\bf 1}^T A_0 \mu <
  {\bf 1}^T A_2 \mu$, null recurrent if ${\bf 1}^T A_0 \mu = {\bf 1}^T
  A_2 \mu$ and positive recurrent if ${\bf 1}^T A_0 \mu > {\bf 1}^T
  A_2 \mu$.
\end{theorem}

By Theorem~\ref{th:StatDistr}, the QBDP does not have an equilibrium
if and only if it has at least one state that is transient or null
recurrent.  Without loss of generality, assume that $n\leq 2H$ (the
case $n > 2H$ can be equivalently derived), which implies
$A\in\real^{H+1\times H+1}$.  Since $A$ is irreducible, one
immediately has that $\mu = \frac{1}{H+1} {\bf 1}$, from which it is
possible to explicitly compute
{\small
\[
\begin{aligned}
{\bf 1}^T A_0 \mu & = \frac{1}{H+1} \sum_{j=0}^{H-1} (H-j) a_{j} \\
{\bf 1}^T A_2 \mu & = \frac{1}{H+1} \sum_{j=H+1}^{n} (j-H) a_{j} .
\end{aligned}
\]}
From Theorem~\ref{th:Transient}, the QBDP does not have an equilibrium
if and only if ${\bf 1}^T A_0 \mu \leq {\bf 1}^T A_2 \mu$ or,
equivalently,
{\small
\[
\sum_{j=0}^{H-1} (H-j) a_{j} \leq \sum_{j=H+1}^{n} (j-H) a_{j} ,
\]}

that, dividing both terms by $a_0$ leads to
{\small
\begin{equation}
\sum_{j=0}^{H-1} \gm{j,1} \leq \sum_{j=H+1}^{n} (j-H) \alpha_{j} .
\label{eq:condition-case-I}
\end{equation}}
This condition is exactly the one that we formulated in the case I of
the Theorem, and has just been shown to be equivalent to the process
being transient on null recurrent.  However, since the QBDP is still
irreducible and aperiodic, a limiting probability exists, which is
given,as in Equation~\eqref{eq:FinalResult1}, by:
{\small
\[
\tilde{\pi}^{(j)} =\lim_{k\rightarrow +\infty} \tilde{\pi}^{(j)}(k) = 0,\,\,\, \forall
j ,
\]}
And this ends the proof of Theorem~\ref{th:Final}.
\begin{remark}
  When condition~\eqref{eq:condition-case-I} strictly applies, the
  numerator of Equation~\eqref{eq:Pi0Unst} is negative.  Since
  Equation~\eqref{eq:FinalResult2} still holds true, the denominator
  of~\eqref{eq:Pi0Unst} will be negative too. It follows that in the
  case of absence of an equilibrium for the QDBP,
  both~\eqref{eq:FinalResult2} and~\eqref{eq:Pi0Unst} return a
  coincident value $\tilde{\pi}^{(0)} > 1$, clearly unfeasible.
\end{remark}

\subsection{Computation of the bound}

As discussed earlier, the steady state probability of meeting the
deadline can be found by computing the first element
$\tilde{\pi}^{(0)}$ of the $\tilde{\pi}$ that solves the equation
$\tilde{\pi} = \tilde{\pi} P$, where $P$ is the infinite transition
matrix in Fig.~\ref{fig:structure} associated with the DTMC
$\mathcal{M}$.  Let us consider a new DTMC whose transition matrix is
given by:
{\small
\begin{equation}
\label{eq:TransMatrix}
P^{'} = \begin{bmatrix} 
b_H & a_{H+1} & a_{H+2} & \dots & a_{n-1} & a_{n} & 0 & \ldots \\ 
b_{H-1} & a_{H} &a_{H+1} &\ldots & a_{n-2} &a_{n-1} & a_n &\ldots \\
0 & a_{H-1}^{'} & a_{H} & \dots & a_{n-3} & a_{n-2} & \ddots & \ldots \\
0 & 0 & a_{H-1}^{'} & a_{H} & \dots & a_{n-3} & \ddots & \ldots \\ 
\vdots & \vdots & \vdots & \vdots &  \ddots & & \ddots & \ddots \\  
\end{bmatrix} ,
\end{equation} }
and $a_{H-1}^{'} = b_{H-1} = a_{H-1} + a_{H-2}+\ldots+a_{0}$. 
\begin{remark}
\label{rem:lump}
The underlying idea is very simple. Consider the DTMC
associated with matrix $P$. The terms on the left of the diagonal are
transition probabilities toward states with a smaller delay than the
current one. By using $P^{'}$ we lump together all these transitions
to the state immediately on the left of the current one. For instance,
if the current state corresponds to $4$ server periods of delay, its
only enabled transition to the left will be to the state associated
with delay $3$.  The effect of deleting the transition toward states
associated with smaller delays is to slow down the convergence toward
small delays, thus decreasing the steady state probability of these
states.
\end{remark}
Let
$\pi$ represent the steady state probability of this system. We can
easily show the following:
\begin{lemma}
  Let $\Gamma$ be a random variable representing the state of the DTMC
  evolving with transition matrix $P$ and $\Gamma^{'}$ be a random
  variable describing the state of the DTMC associated with the
  transition matrix $P^{'}$.  If both DTMC are irreducible and
  aperiodic, then at the steady state $\Gamma^{'}$ has a first order
  stochastic dominance over $\Gamma$: $\Gamma^{'} \succeq \Gamma$,
  according to Definition~\ref{def:Succ}. Therefore, for the first
  element of the steady state probability, we have $\tilde{\pi}^{(0)}
  \geq \pi^{(0)}$.
\label{lem:lemm1}
\end{lemma}

\begin{proof}
  The proof is omitted for the sake of brevity (see~\cite{TR}).
  
\end{proof}

In view of this Lemma, we can concentrate on the system associated to
the transition matrix $P'$.  In such a case, we immediately derive
that the equilibrium condition $\pi = \pi P'$ produces the following
recursion:
{\small
\begin{equation}
\label{eq:RecursivePiClosedForm}
\begin{aligned}
  \pi^{(1)} & = \sum_{j=2}^{n} \alpha_j \pi^{0} , \\
  \pi^{(l)} & = \left( 1+\sum_{j=2}^{n} \alpha_j \right) \pi^{(l-1)} -
  \!\!\!\!\!\! \sum_{j=2}^{\min(n,H+l-1)} \!\!\!\!\!\! \alpha_{j}
  \pi^{(l-j)},
\end{aligned}
\end{equation}}
where the equalities hold for $\forall l > 1$.  This equations, as
well as $P'$, have been respectively derived
from~\eqref{eq:RecursivePi} and $P$ by imposing $H = 1$. In such a
situation, the following theorem holds.

\begin{theorem}
\label{th:FinalClosedForm}
Consider a QBDP described by the transition probability
matrix~\eqref{eq:TransMatrix}, in which both $a_n$ and $a'_{H-1}$
differ from zero. Assume that the matrix $W$ in~\eqref{eq:MatW} has
distinct eigenvalues after imposing $H = 1$. Then, there exists a
limiting probability distribution given by
{\small
\begin{equation}
  \label{eq:FinalResultClosedForm}
  \begin{aligned}
    \pi^{(0)} & = \lim_{k\rightarrow +\infty} \pi^{(0)}(k) = \max\{1 - \sum_{j=2}^{n} (j-1) \alpha_j,\, 0\} =  \\
    & = \max\{1 - \sum_{j=2}^{n} (j-1) \frac{a_j}{a_0},\, 0\} ,
  \end{aligned}
\end{equation}}
while the generic terms $\pi^{(j)}$, with $j>0$, are given
by~\eqref{eq:RecursivePiClosedForm}.
\end{theorem}
\begin{proof}
  The proof follows immediately from the fact that $H = 1$ implies
  that $\beta_1 = 1$ is the only unstable eigenvalue if the QBDP has
  an equilibrium, i.e., $\mathcal{B}_s$ of Theorem~\ref{th:Final}
  comprises all the eigenvalues except $\beta_1 = 1$.  Hence, by
  considering~\eqref{eq:EigProdSimply} for $H = 1$, the proof follows
  immediately.
\end{proof}

We complete the section with a remark.  The first one is on the
intuitive meaning of the result just proposed.  Consider a DTMC with
transition matrix as in Fig.~\ref{fig:structure} and assume for
simplicity $n=4$ and $H=1$. The analytical bound in
Theorem~\ref{th:FinalClosedForm} is given by:
{\small
\[
\begin{array}{ll}
\pi^{(0)}&= 1 - 3 \alpha_4 - 2 \alpha_3 - \alpha_2  = 1 - 3 \frac{a_4}{a_0} - 2 \frac{a_3}{a_0} - \frac{a_2}{a_0}
\end{array}
\]}
In the computation of the steady state probability $\pi^{(0)}$ we have
to consider every possible transition to the right (i.e., increasing
the delay) that the system can make. For each of them, we compute the
ratio between the probability of taking the transition and the
aggregate probability of moving to the left (decreasing the delay).
In the final computation each of this ratio has a state proportional
to the delay introduced. In our example, $a_4$ corresponds to three
steps to the right and is weighted by the factor $3$.

The application of this result to our context can be formalised in the
following:
\begin{corollary}
  \label{cor:FinalClosedForm}

  Consider a resource reservation used to schedule a periodic task and
  suppose that the QBDP produced respects the assumption in
  Theorem~\ref{th:Final}.  Then the probability of respecting the
  deadline is greater than or equal to:
{\small
\begin{equation}
\pi^{(0)}= 1 - \sum_{j=2}^{n}(j-1) \frac{U'_\Delta(N+j-1)Q^s)}{ \sum_{h=0}^{N-1}U'_\Delta(hQ^s)}
\end{equation}}

\end{corollary}
This corollary descends from the following facts: 1) the DTMC
described by the matrix $P$ in Fig.~\ref{fig:structure} is a
conservative approximation of the system, 2) Lemma~\ref{lem:lemm1}
provides an analytically tractable approximation of the DTMC with
transition matrix $P'$, 3) Theorem~\ref{th:Final} and
Theorem~\ref{th:FinalClosedForm} contain the analytical bounds.

\section{Experimental validation}
\label{sec:experiments}

We have validated the presented approach in two different ways.
First, we have computed the probabilistic deadline using synthetic
distributions, to compare accuracy and efficiency of the analytic
bound against several other methods and to assess the impact of the
scaling factor $\Delta$ (Eq.~\eqref{eq:uh}) and of the bandwidth.
This set of experiment reveals a very good performance of the bound
for appropriate choices of the scaling factor $\Delta$. Its very low
computation time allows one to select the best choice of $\Delta$ by
testing a number of alternative choices.
The tightness of the bound improves when the bandwidth is
sufficient to achieve an acceptable real--time behaviour for the
application.

In a second set of experiments, we have evaluated the method on a real
robotic application, for which the mathematical assumptions underlying
the model do not apply strictly.  The results produced are obviously
approximate. Still, the good quality of the approximation makes an
interesting case for the practical applicability of the methodology.

\subsection{Synthetic Distributions} 
\label{sec:experiments-synthetic}
We report the results of the comparison between the numeric solution
resulting from Theorem~\ref{th:Final} and discussed in
Remark~\ref{rem:companion} ({\tt companion}), the analytic
approximated bound in Corollary~\ref{cor:FinalClosedForm} ({\tt
  analytic}) the Cyclic Reduction algorithm~\cite{bini2005numerical}
({\tt CR}) and the bound developed by Abeni et al.~\cite{Abe12} ({\tt
  gamma}).  We have chosen {\tt CR} after a selection process in which
several algorithms for the solution of general QBDP problems and
implemented in the SMCSolver tool--suite~\cite{bini2012smcsolver} were
tested on a set of example QBDPs derived from our application. The
{\tt gamma} algorithm is an approximate bound specifically tailored to
the analysis of probabilistic guarantees for resource reservations, so
it was considered as as a perfect match for our {\tt analytic} bound.
The different algorithms have been implemented in C++ in the
PROSIT~\cite{WATERS2014} tool. PROSIT can be used for analysis and for
synthesis purposes (as shown in Section~\ref{sec:example}).  When the
tool is used for analysis, the user specificies activation period and
deadline, parameters of the RR ($Q^s$ and $T^s$), distribution of computation and
inter--arrival times and solution algorithm.  When the tool is queried
in this way, it computes the distribution of the task
response times and hence the probability of meeting the deadline.

As a representative sample of our findings, we report below the
results obtained for a periodic task with period $T=100ms$ and random
execution time.  The computation time was distributed according to a
beta distribution: $P\left\{C=c\right\}=f_U(c) = J(\alpha,\beta)
c^{\alpha-1} \left(1-c\right)^{\beta-1}$, with support (i.e., the
validity range for the random variable) $c \in \left[0,
  99500\right]$~$\mu$s, with $\alpha = 2$ and $\beta=7$
($J(\alpha,\beta)$ is a normalisation constant).  The beta
distribution is interesting because it is unimodal and has a finite
support, which make it a good fit to approximate the behaviour of a
large number of real--time applications.

\noindent
{\bf Effect of $\Delta$.}
A first set of experiments was to evaluate the impact of the $\Delta$
scaling factor. We considered two possible values for the
reservation period: $T^s=\frac{1}{4}P=25ms$ and
$T^s=\frac{1}{2}P=50ms$.  The budget was chosen equal to $Q^s =
0.45 T^s$ with a bandwidth $B=45\%$. Figure~\ref{fig:var-delta}
shows the results for the probability $\pi^{(0)}$ of respecting the
deadline achieved for different values of $\Delta$ (chosen as a
sub--multiple of $Q^s$).
\begin{figure}[t]
  {\includegraphics[height=7.5cm,width=\columnwidth]{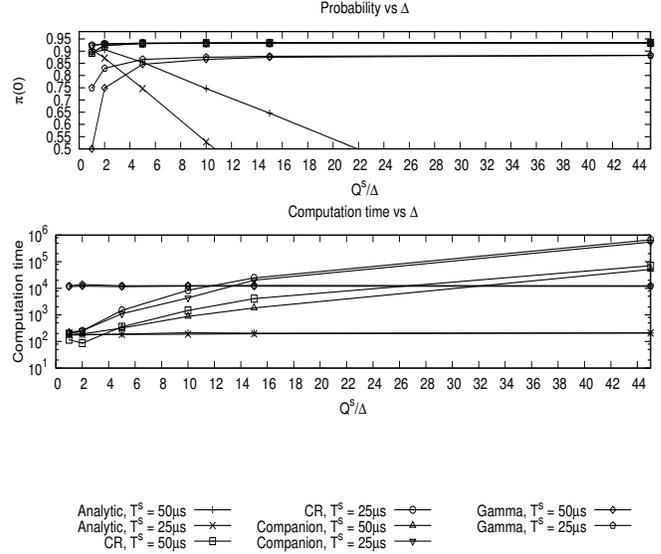}}
  \caption{Impact of the scaling factor $\Delta$ on the accuracy of the computed probability and on
    the computation time}
  \label{fig:var-delta}
\end{figure}
In accordance with our expectations, {\tt CR} and {\tt companion}
produce almost the same result in term of probability (differences are
from the $6^{th}$ digit) and the probability changes monotonically
with $\Delta$. For example, for $T^s=50ms$ the value of the
probability is $0.89$ for $\Delta = Q^s$ (the coarsest possible
granularity), while it is $0.93$ for $\Delta = Q^s/45$.  The reason
for this decrease is obvious since re--sampling introduces a
conservative approximation and the error is larger for increasing
granularity.  For both {\tt CR} and {\tt companion}, the computation
time changes with $\Delta$ in a substantial way. For example, for {\tt
  CR} and for $T^s=50ms$, it is $182ms$ at $\Delta=Q^s$ and
$56.179ms$ at $\Delta = Q^s/45$.  In this run of experiments, the
computation time of the {\tt companion} algorithm is slightly smaller
than the one reported using {\tt CR}, but the results are too close to
claim a clear dominance.

For the {\tt analytic} bound the computed probability is not always
monotonic with $\Delta$.  In our example, for $T^s=50ms$ the
probability grows moving from $0.892$ at $Q^s$ to $0.906$ at $Q^s/2$,
and then decreases, finally becoming $0.012$ at $Q^s/45$.  Sharper
changes can be observed for other distributions.  The reason is that
in the {\tt analytic} bound we have two distinct effects (which play
in opposite directions). On the one hand, if we reduce $Q^s$ we have
the same conservative approximation effect as for {\tt CR} or for any
other numeric method.  On the other, as explained in
Remark~\ref{rem:lump}, lumping together all backward transitions
reduce the recovery of the error when the computation demand is
smaller than the allocated bandwidth. In this example, the first
effect determines the growth of the probability when going from
$\Delta=Q^s$ to $\Delta=Q^s/2$; the second effect determines the
decrease of the probability form $Q^s/2$ onward. The probability
computed by {\tt analytic} is very close to the one of the numeric
algorithm it derives from ({\tt companion}) for $\Delta = Q^s/2$,
while the computation time is several orders of magnitude below.  In
our experience with different distributions (both synthetic and
experimental) the choice of $\Delta=Q^s/2$ has consistently produced
an acceptable performance.  The {\tt gamma} bound shows an
intermediate performance between numeric methods and the analytic
bound both for the accuracy and for the computation time.

\noindent {\bf Behaviour with changing bandwidth.} In order to compare
the accuracy of the {\tt analytic} method against the numeric solutions
({\tt CR}) for different bandwidths, we considered a task with the
activation and scheduling parameters as in the experiments reported
above. The budget $Q^s$ was changed so that the resulting bandwidth
ranged in $\left[35\%,\,60\%\right]$.  The granularity $\Delta$ was
fixed for {\tt CR} to a small value ($50\mu$s) to achieve a good
approximation and to $\Delta=Q^s/2$ for the {\tt analytic} solution.

\begin{table}
\caption{Probability for different bandwidth and $\Delta=50us$}
\label{tab:optimisation1}
{\small
\centerline{\begin{tabular}{cccccc}
\hline
Bandwith & 35\% & 40\%  & 45\% & 50\% & 60\%\\
\hline
\hline
Analytic Bound &0.602 & 0.809 & 0.906 &0.956  & 0.991\\
\hline
Cyclic Reduction & 0.773 & 0.878 & 0.929 & 0.965 & 0.992 \\
\hline
\end{tabular}}}
\end{table}

The results reported in Table~\ref{tab:optimisation1} show an
important gap between {\tt analytic} and {\tt CR} for small values of
the bandwidth. The gap is significantly reduced for bandwidth greater
than $45\%/50\%$. Smaller values of the bandwidth produce a
probability level below $0.8$, which is not acceptable for most
real--time applications.  The reason for the improvement of the analytic bound
when the bandwidth increases is probably due to the fact that the
system recovers more easily from large delays and this alleviates the
impact of the conservative simplifications that underlie the analytic
model.

\subsection{Real application} 
As a real test case, we have considered
a robotic vision programme that identifies the boundaries of the lane
and estimates the position of a mobile robot a using a web--cam mounted on the chassis
of the robot~\cite{DBLP:journals/tim/FontanelliMRP14}. The
computation was carried out using a Beagle Board
(\url{www.beagleboard.org}) running Ubuntu. The version of the Kernel
used (3.16) supports RR scheduling (under the name
of SCHED\_DEADLINE policy) alongside the standard POSIX real--time
fixed priority policies (SCHED\_FIFO and SCHED\_RR).

The robot executed $30$ different paths across an area delimited by a
black line. For each run, we have captured a video stream containing
the line.  The data sets roughly consisted of $2500$ frames each and
were later used for multiple off--line execution of the vision
algorithm.  A first group of ten executions for each data set was with
the algorithm executed in a task running alone and scheduled with the
the maximum real--time priority (99 for SCHED\_FIFO).  This allowed us
to collect statistics of the computation time associated with the data
set.  In a second group of executions, we have replicated a real--life
condition. The vision algorithm was in this case executed in a
periodic task processing a frame every $T=40ms$.  The task was
scheduled using SCHED\_DEADLINE, with
server period $T^s = 20ms$ and with different choices of the bandwidth
in the range $[35\%, 60\%]$.  For each data set and for each choice of
the bandwidth, we repeated ten executions recording the probability of
deadline miss.
The probability averaged through the $10$ execution was compared with
the one that found using the PROSIT tool, executed with different
solution methods and with the distribution estimated from the data set
as input. In Figure~\ref{fig:3000_best}, we report the CDF
distributions of the difference between the two probabilities for
three representative choices of the bandwidth.  The symbol
$\Delta_{{\tt Analytic}}$ denotes the difference obtained using the
analytic method (with different choices of the scaling factor
$\Delta$), while $\Delta_{{\tt CR}}$ denotes the difference obtained
using the cyclic reduction QBDP solver, with $\Delta$ set to $50\mu
s$.  The three levels of bandwidth shown in the three sub--plots
produced different probability of meeting the deadline. For bandwidth
equal to $40\%$, this probability ranged in $[75\%, 97\%]$. The range
was $[90.5\%, 99\%]$ for bandwidth equal to $50$\% and it was
$[95.2\%, 100\%]$ for bandwidth equal to $60$\%.

As we observe in the plot, the numeric algorithm (CR) produces an
error between $-3$\% and $1$\% for all the three values of the
bandwidth.  For the analytic bound, in this specific case, the most
convenient choice was to set the scaling factor $\Delta$ to $Q^s$ (in
other cases we found a better performance for smaller values).  The
bound is evidently less accurate, but: 1. it remains below $5\%$ at
least $85\%$ of the times even in the most challenging scenario (small
bandwidth), 2. is reduced to below $2\%$ for higher values of the
bandwidth.

\begin{figure}[t]
\centering
\begin{tabular}{c}
\includegraphics[height=7.5cm,width=0.97\columnwidth]{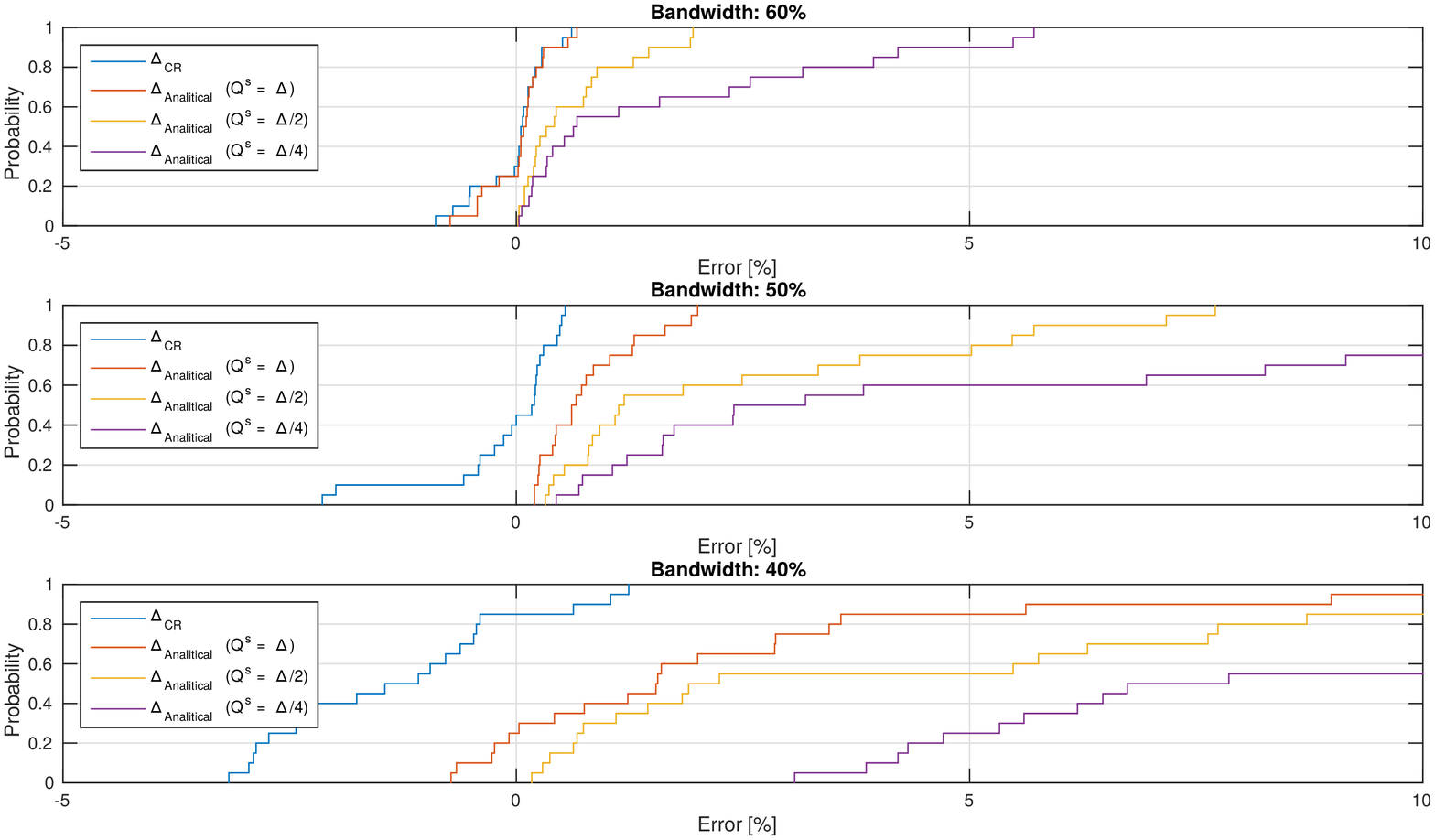} \\
\end{tabular}
\caption{Distribution of the difference between the experimental probability and
the one found with PROSIT tool.}
\label{fig:3000_best}
\end{figure}

We observe that the vision algorithm iteratively builds upon
previous results to produce the estimate. This introduces a strong
correlation structure in the process that disrupts the assumptions
required for an exact application of the method. In addition, the
execution on a ``real'' operating system comes along with an
inevitable amount of un--modelled overhead.  Still, the level of
approximation that we have reported could be acceptable in most cases.
Similar software applications (video--encoding and decoding) were
analysed in a previous work~\cite{ecrts2012} with similar conclusions.
Clearly, we are not claiming any generality for this fact. We are
aware that for other applications dropping the time dependency and the
correlation structure of the computation time process could produce
very large errors in the estimation of the probability.  As reported
in the related work, this is a very active research area that is
likely to attract the attention of different researchers in the
forthcoming years.

\subsection{Discussion}
\label{sec:discussion}

In our first conference paper~\cite{etfa2012}, we derived a model for
the evolution of a RR scheduled real--time task. The model was shown
to be a QBDP and was solved using the
simple numeric algorithm proposed by Latouche and
Ramaswami~\cite{LatoucheR87}. An important limitation of the model was
its pessimism due to the fact that it neglected the budget shared
between adjacent jobs. For instance, in the example in
Figure~\ref{fig:example}, the model would ignore the budget used by
the second job in the fourth reservation period.  In a later
work~\cite{ecrts2012}, the same model was instantiated to the sub--case
of periodic tasks, it was further simplified in a conservative
direction and then used for the computation of an analytic bound.

In the present paper, we start from the more accurate model introduced
by Abeni and Buttazzo back in 1998~\cite{Abe98-3}, and we instantiate
it to the case of periodic tasks (Section~\ref{sec:full_model}).  We
introduce the scaling factor $\Delta$
(Section~\ref{sec:simplification}) obtaining, once again, a QBDP.
When the model is used for numeric computations, the $\Delta$
parameter allows us to decide the degree of
pessimism introduced in the analysis. If we set $\Delta=1$, we obtain
a close approximation of the actual behaviour of the task. If we set
$\Delta = Q^s$, we recover the conservative model used in our previous
work~\cite{etfa2012}.  As shown in Figure~\ref{fig:var-delta},
very different trade--offs between 
computation time and accuracy of the probability result from different
choices of $\Delta$.

The key contribution of this paper is found by applying the same type
of analytic reasoning as in~\cite{ecrts2012}, but with a few
substantial differences in the final result.  Indeed,
Theorem~\ref{th:Final} contains an exact formula for the computation
of the steady state probability of meeting the deadline, which is used
as a basis for a novel numeric algorithm with competitive performance
with respect to the state of the art.  On the contrary, the key result
of~\cite{ecrts2012} is an analytic bound which can sometimes be very
conservative.  The same bound is rediscovered in this paper
specialising Theorem~\ref{th:Final} to a conservative approximation of
the model (see Theorem~\ref{th:FinalClosedForm}).  Once again, we can
take advantage of the configuration options offered by $\Delta$ to
refine the precision of the result. As shown in Figure~\ref{fig:3000_best},
the choice $\Delta = Q^s$ (which applies the model proposed
in~\cite{ecrts2012}) is not guaranteed to be the best one in all cases.
Therefore, the generalisation shown in this paper is relevant both
from the theoretical and from the practical point of view.

\section{Probabilistic Quality Optimisation}
\label{sec:example}
In order to show a practical application of our approach, we have
considered a situation where a single computing board (e.g., a video
server, or a set--top box) is used to process (in real--time) multiple videos at the
same time.
This example is based on two different videos (encoded with a bit--rate of $600$Kb/s): the first one,
``BridgeClose'', displays a bridge with occasional people coming through
(so, it is characterised by a single, almost static scene with slow movements)
and comes from a
public data base (\url{http://trace.eas.asu.edu/yuv/index.html}); the second
video (``ufo''), instead, is a movie trailer
(freely available at \url{http://www.theufo.net} - trailer 1) characterised
by frequent scene changes and rapid movements.

One of the best known ways to evaluate the quality of a video is the
Peak Signal to Noise Ratio (PSNR), which is computed comparing
pairwise the frames of the original raw video and of the one obtained
after encoding and decoding it~\cite{Klaue03evalvid,psnr-tools}.  This
metric can be evaluated considering a video player implemented as a
periodic real--time task.  If a job misses its deadline, the video
frame is not played back but it is decoded (to allow the incremental
decoding of the frames that follow).  In this case, the behaviour of
most players is to fill--in the ``hole'' by simply repeating the last
decoded frame.  This is perceived by the user as a reduction in
quality, which is well reflected in a degradation of the PSNR.
\begin{figure}[t]
\centering{
\begin{tabular}{c}
\includegraphics[width=0.9\columnwidth,height=4cm]{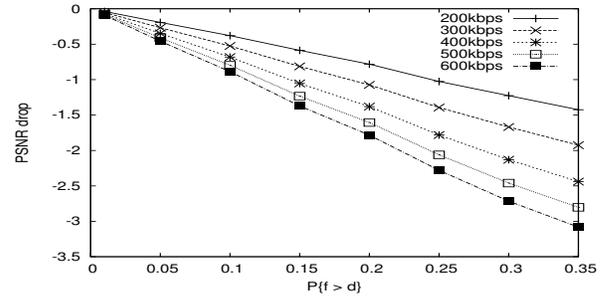}
\end{tabular}}
  \caption{PSNR degradation  as a function of the deadline miss probability for
``BridgeClose'' video.}
  \label{fig:psnr2}
\end{figure}
This is visible in Fig.~\ref{fig:psnr2}, where we show the quality as
a function of the probability of deadline miss for the first video. 
This plot has been created using the PSNR--TOOL
software~\cite{psnr-tools}.

The PSNR was interpolated by a line with slope $8.9$ for
``BridgeClose'' and $42.051$ for ``ufo''. This difference is explained
by the different nature of the movies (static the former, and dynamic the latter).
Both movies
were decoded using a player executed by a periodic task and scheduled
by the SCHED\_DEADLINE policy.
The distributions of the execution times were recorded on a notebook
powered by an Intel Atom Processor, and the resulting CDFs are shown in
Fig.~\ref{fig:cdc-streams}.
\begin{figure}
\centering{\includegraphics[width=0.9\columnwidth, height=4cm]{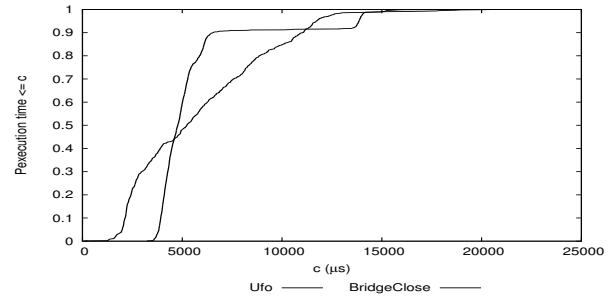}}
\caption{Cumulative Distribution Functions for the execution of the decode for the two streams.}
\label{fig:cdc-streams}
\end{figure}
\begin{table}
\caption{Results of Probabilistic Optimisation}
\label{tab:optimisation}
{\small
\centerline{\begin{tabular}{|c|c|c|c|c|}
\multicolumn{5}{c}{Cyclic Reduction -- Computation time:753801758$\mu$s}\\
\hline
Task & Opt. Budget  & Estim. Prob. & Exact Prob. & Quality\\
\hline
BridgeClose & 3000us & 0.7427 &0.743592.  & 39.65\\
Ufo &6449us & 0.9995 & 0.9995 & 41.58 \\
\hline
\multicolumn{5}{c}{Analytic Bound  -- Computation time:114524$\mu$s}\\
\hline
Task & Opt. Budget  & Estim. Prob. & Exact Prob. & Quality\\
\hline
BridgeClose &3462us & 0.7392 & 0.8292 & 40.50\\
Ufo & 3997us & 0.8732 & 0.9138 & 37.98 \\
\hline
\end{tabular}}}
\end{table}

The problem considered here was to find an optimal allocation of
bandwidth between the different tasks.  To this end, we have used
the synthesis abilities of PROSIT. When PROSIT is used for synthesis,
the user specifies for each task: 1) activation period and deadline,
2) reservation period, 3) distribution of the computation time 4) solution
algorithm for the probabilistic guarantees,
5) quality as a function of the probability of
meeting the deadline and 6) constraints on the minimal value of the
quality.  The quality of the
different tasks can be combined into global quality metrics.  In this
particular example, we have used the infinity norm metric: assuming
$f_i$ as the quality of the $i^{th}$ task, the cost function to
maximise over the budget $Q^s_1$ and $Q^s_2$ is $\max_i \min f_i$.
For each candidate choice of $Q^s_i$ the tool evaluates the steady
state probability using different solvers for probabilistic
guarantees.  The optimal solution is found by a bisection algorithm,
which uses repeated calls to the algorithm for the computation of the
probability.
As a solver for the probability computation
we have implemented {\tt analytic} (with $\Delta=Q^s/2$) and {\tt CR}
(with $\Delta=50$~$\mu$s).

Choosing $30$~ms for the activation period (corresponding to
$33$~fps), setting the server period to $10$~ms, and restricting the
total bandwidth available to $95\%$ (to leave some room for other
applications), the tool produces the results in
Table~\ref{tab:optimisation}.  We identified empirically the minimum
acceptable PSNR as $39$ for ``Ufo'' and $31$ for ``BridgeClose''.
These values were codified as constraints in the optimisation problem.
In both cases, the algorithm identified a sub--optimal solution,
because the probability evaluated by the solvers is only a lower
bound.  We re--evaluated the exact probability for each of the
sub--optimal assignment of budgets using the {\tt CR} solver
with $\Delta=1$ (which produces the exact computation of the
probability, within the limits of numeric errors). This allowed us to
compare the actual quality attained by the optimisation algorithm in
the two different configurations.
Because the optimiser maximises the worst performance of
the two tasks, the algorithm tends to equalise the QoS achieved by the
tasks for the optimal budget. For both solvers, the optimal solution
assigns a larger bandwidth (almost $64\%$ for the {\tt CR} and almost
$40\%$ for the {\tt analytic}) to the ``Ufo'' stream; this is because
its quality degrades more quickly with the probability of meeting the
deadline for ``Ufo'' than for ``BridgeClose''. In this example, the
use of the analytic bound produces an optimal value $37.98$ which is
only $4\%$ away from the value obtained with cyclic reduction, but the
computation time (evaluated on an Intel Core i7 with $16GB$ of RAM) is
four orders of magnitude below.

\section{Conclusions and Future Work}
\label{sec:conclusions}

In this paper, we have considered the problem of probabilistic
guarantees for RR scheduled soft real--time periodic tasks.  We have
shown that the evolution of the system can be modelled as a QBDP. The
probability of meeting the deadline amounts to the computation of the
steady state probability of this process.  We have shown how this is
possible by numeric means with different performance/accuracy
tradeoffs. We have also shown an analytical bound and offered a
comprehensive validation of these results by experiments and
simulations.

The gap between the analytic bound and precise numeric solution
narrows down when the task is required to meet the deadline with a
high probability (e.g., more than $80\%$).  For this reason, the
analytic bound appears as a very promising option to solve QoS
optimisation problems involving multiple tasks, when the QoS is a
function of the probability for the task to meet its deadline and an
acceptable level of performance is required to all tasks.  In these
cases, the frequent calls to the solver to identify the optimal
allocation of resources, such as are required by branch and bound or
dichotomic search optimisation, can lead to substantial reduction of
the computation time when the analytic bound is used in the face of an
acceptable distance from the optimal solution.

\noindent
{\bf Future work} In our future work, we will investigate further on
the connection between QoS and probabilistic deadlines in several
application domains, we will extend our analysis and the application
of our methods to the case of applications based on multiple tasks and
to the case of computation time that is not i.i.d.

\bibliographystyle{IEEEtran}
\bibliography{retis}

\end{document}